\documentclass[11pt,oneside,reqno]{amsart}

\usepackage{amsfonts}
\usepackage{pdfsync}
\usepackage{amsmath}
\usepackage{amssymb}
\usepackage{amsthm}
\usepackage{caption}
\usepackage{enumerate} 
\usepackage{dcolumn} 

\newtheorem{theorem}{Theorem}[section]
\newtheorem{lemma}[theorem]{Lemma}

\newtheorem*{remarks}{Remarks}

\usepackage{color}

\def\ii{\'\i}

\begin{document}
\title[Excess charge for the Hartree Model]{\textbf{Analytic bound on the excess charge for the Hartree Model}}

\author[Benguria]{Rafael~D.~Benguria$^1$}

\author[Tubino]{Trinidad Tubino$^2$}

\address{$^1$ Instituto  de F\'\i sica, Pontificia Universidad Cat\' olica de Chile,}
\email{{rbenguri@fis.puc.cl}}

\address{$^2$ Facultad de F\'\i sica, Pontificia Universidad Cat\' olica de Chile,}
\email{{trinitubino@uc.cl}}

\begin{abstract} 
We prove an analytic bound on the excess charge for the Hartree equation in the atomic case.
\end{abstract}

\maketitle

\section{Introduction} \label{intro}
The ability of atoms and molecules to catch or give electrons is one of the fundamental pillars of chemistry and biology, 
and it has somewhat attracted attention since the late XVIIIth century with the birth of modern chemistry. Given a neutral atom, 
one can extract one electron by providing enough energy. The energy needed to remove one electron is called its {\it ionization energy} 
(e.g., the energy needed to extract the electron of the Hydrogen atom in its ground state is approximately $13.6$ eV). One can continue 
this process of removing electrons by providing increasing amounts of energy (a fact that seems intuitively obvious but yet to be proven), 
up to stripping the atom completely of all its electrons. On the other hand, the question up to what extent one can add electrons to a neutral atom
 is not yet completely understood. After the introduction of the Schr\"odinger equation in 1926, Hans Bethe was the first to study the possibility of 
 having the negative Hydrogen ion $H^{-}$. Using the variational principle of E. Hylleraas, he proved the existence of a bound state for $H^{-}$ 
 \cite{Be29}, and computed quite accurately the ionization energy of the first electron to be $17$ kcal/mol (i.e., approximately $0.74$ eV, whereas
 its present value is approximately $0.75$ eV). With  the discovery of Rupert Wildt \cite{Wi39} that the presence of $H^{-}$ in the solar atmosphere  
 is  the main cause of its opacity (in the visible, but specially in the infrared)  there was a renewed interest in 
 the spectral properties of the $H^{-}$ anion (see, in particular, the review articles \cite{Ch44}, \cite{Hy64}, and the monograph \cite{He86}, pp. 404--ff). 
 A rigorous proof that $H^{-}$ has only one bound state (i.e., the ground state) and no singly excited states had to wait approximately fifty years, 
 until the work of Robert Hill \cite{Hi77a,Hi77b}. For a review of the physics literature on the Hydrogen anion up to 1996, we refer the reader to the article of Rau \cite{Ra96}. 

In the last half a century there has been a big effort in trying to determine the maximum number of electrons an atomic nucleus or a molecule can bind. 
These efforts come from three different fronts. From one side, in experimental physics, there has been an intensive search  for {\it dianions} 
(i.e., doubly ionized atoms and molecules). On another front, there has been an intensive numerical search for the possibility of stable atomic dianions \cite{Ho98} 
giving strong evidence that stable atomic dianions do not exist. Finally, in mathematical physics there has been an interesting quest to solve this problem. 
Based on the knowledge  coming from these three fronts, we expect that an atom of atomic number $Z$ can bind at most $Z+1$ electrons, 
while a molecule of $K$ nuclei of total nuclear charge $Z$ can bind at most $Z+K$ electrons.

There are two main ingredients involved in this question: the fact that the maximum number of electrons an atom of nuclear charge $Z$ can bind is at least $Z$ 
(i.e., that neutral atoms do exist) is very much related to the mathematical properties of the Coulomb interaction between charged particles. 
In particular, a crucial role is played by Newton's theorem. On the other hand, the expected fact that at most $Z+1$ electrons can be  bound has to do 
with Pauli's Exclusion Principle (i.e., more precisely with the fact that electrons obey Fermi statistics). 

We have seen above that the first rigorous results (of Bethe and Hill) on negative ions dealt with $H^{-}$. Concerning the more general situation of 
$N$ electrons and $K$ (usually fixed)  nuclei interacting via Coulomb potentials the first results were obtained by Zhislin (see \cite{Zh60,Zh71}) 
who proved that below neutrality (i.e., when the total number of electrons is strictly less than the total nuclear charge) the corresponding 
Hamiltonian in non--relativistic quantum mechanics has an infinite number of bound states, whereas at neutrality or above it, the number of 
possible bound states is at most finite. Upper bounds on the number of bound states for bosonic matter above neutrality were obtained later in 
\cite{BaLeLiSi93,Rus994}. At the beginning of the 80's, Ruskai and Sigal \cite{Rus981,Rus982,Sig982,Sig984}, using the IMS localization formula and appropriate 
partitions of unity obtained the first actual upper bounds on the maximum number of electrons an atom or molecule can bind. 
In 1983, Benguria and Lieb \cite{BeLi83} proved that the Pauli principle is crucial when considering the problem of the maximum 
number of electrons an atom can bind. In fact, they proved that $N_c(Z)-Z \ge c Z$ as $Z \to \infty$, where $c$ is obtained by solving 
the Hartree equation (which is equation (\ref{eq:HF2}) below). Here we denote by $N_c(Z)$ 
the maximum number of electrons a nucleus of charge $Z$ can bind. Then, Baumgartner \cite{Bau984} solved numerically the Hartree equation to find $c \approx 0.21$. 
Later,  Solovej \cite{So90} obtained an upper bound which showed that $N_c(Z)=1.21\,  Z$ is the appropriate asymptotic formula for large $Z$. In 1984, 
Lieb obtained the simple upper bound $N_c(Z) < 2 Z + K$ independently of statistics \cite{Lie984a,Lie984b}, and Lieb, Sigal, Simon and Thirring proved 
that fermionic matter is asymptotically neutral (i.e., $N_c(Z)/Z \to 1$ as $Z$ goes to infinity \cite{LiSiSiTh84,LiSiSiTh88}). In 1990, Fefferman and Seco \cite{FeSe90} 
obtained a correction term to this asymptotic neutrality, namely they proved that
$N_c(Z) \le Z + c Z^{1-\alpha}$,
for some constant $c$, with $\alpha=9/56$. The proof of this result was later simplified by Seco, Sigal and Sovolej \cite{SeSiSo90} who established a 
connection between the ionization energy and the excess charge $N_c(Z)-Z$, and estimated asymptotically the ionization energy. 
More recently P.-T.~Nam \cite{Nam012} proved that the maximum number $N_c$ of non-relativistic electrons that a nucleus of charge
 $Z$ can bind is less than $1.22 \, Z+3 \, Z^{1/3}$, which improves  Lieb's upper bound $N_c<2Z+1$  when $Z\ge 6$. 

The conjecture we mentioned at the beginning to the effect that the excess charge $N_c(Z)-Z \le 1$ for an atom is still open. 
However, for semiclassical models (including the Thomas--Fermi model and its extensions, the Hartree--Fock theory, and others) 
there are sharper results. It was proven by Lieb and Simon (\cite{LiSi73,LiSi977b}) that $N_c(Z) = Z$ for the Thomas--Fermi model, 
whereas for the gradient correction (i.e., for the Thomas--Fermi--Weizs\"acker model) Benguria and Lieb proved that $N_c(Z)-Z \le 1$. 
In 1991, Solovej \cite{So91} proved that $N_c(Z)-Z \le c$ for some constant $c$ for a reduced Hartree--Fock model. 
Finally, Solovej in 2003 \cite{So03} proved a similar bound for the full Hartree--Fock model. In the last few years there have been
several articles on the excess charge of different models (see, e.g., \cite{ChSi020,FrNaVa018a, FrNaVa018b, FrNaVa018c,LeLe013})
We also refer the reader to the  monograph of Lieb and Seiringer \cite{LiSe010}, chapter 12, 
for a more complete summary on the maximum ionization.  

\bigskip
\bigskip
In this manuscript, using Nam's technique, we prove an analytic bound on the excess charge for the solution of the Hartree equation. 
Our main result is given in Theorem 3.1 below, where we prove that $N<1.5211 \, Z$. In his original article \cite{Nam012} Nam showed that 
his technique only gives better results for fermions. In view of our result we expect that it could also give better results for an atomic system of 
$N$ bosons.
%
%

\bigskip
\bigskip
We dedicate this paper to Elliott Lieb, in admiration, for his many outstanding 
contributions in Physics,  Analysis and Mathematical Physics.

\bigskip
\bigskip

\section{The Hartree Functional and the Hartree equation}

\bigskip
\bigskip
The Hartree atomic model is defined by the energy functional, 

\bigskip
\begin{equation}
{\mathcal E}[\psi] = \int_{\mathbb{R}^3} (\nabla \psi)^2 \, dx - \int_{\mathbb{R}^3} \frac{Z}{|x|} \, \psi^2 \, dx + \frac{1}{2}  \int_{\mathbb{R}^3} \int_{\mathbb{R}^3} \psi^2(x) \frac{1}{|x-y|} \psi^2(y) \, dx \, dy.
\label{eq:HF1}
\end{equation}

\bigskip
\noindent
This functional is defined for functions  $\psi \in H^1(\mathbb{R}^3)$. Since $\psi \in H^1(\mathbb{R}^3)$, it follows from 
Sobolev's inequality that $\psi \in L^2(\mathbb{R}^3)\cap L^6(\mathbb{R}^3)$, and one can readily check that the second and
 third integral of (\ref{eq:HF1}) are finite. Using the direct calculus of variations one can prove that there is a minimizer of 
 ${\mathcal E}[\psi]$ in $H^1(\mathbb{R}^3)$ (one can obtain the existence of solutions  directly from \cite{BeBrLi981, Lie981}, by setting 
 $p=5/3$ and $\gamma=0$, see also \cite{Bad979,LiSi74,LiSi977a, Stu973,Stu975}). It also follows from \cite{BeBrLi981, Lie981} that the minimizer $\psi$ is such that 
 $\int_{\mathbb{R}^3} \psi^2 \, dx < \infty$. Using the convexity of  ${\mathcal E}[\psi]$ in $\rho=\psi^2$ 
 (see, e.g., \cite{BeBrLi981}, Lemma 4, or \cite{LiLo001}, Theorem 7.8, p. 177), it follows that the minimizer is unique.
 Since the minimizer is unique and the potencial $V(x)= Z/|x|$ is radial (atomic case) we have that the minimizer 
 $\psi(x)$ is radially symmetric. Moreover, the minimizer $\psi$ satisfies the Euler equation (in this case known as the {\it Hartree equation}), 

\begin{equation}
-\Delta \psi = \phi(x) \psi,
\label{eq:HF2}
\end{equation}

\bigskip
\noindent
where the potential $\phi(x)$ is given by

\begin{equation}
 \phi(x) = \frac{Z}{|x|} -   \int_{\mathbb{R}^3}  \frac{1}{|x-y|} \, \psi^2(y) \, dy.
\label{eq:HF3}
\end{equation}

\bigskip
\bigskip
In what follows we need to look at the components of the energy and their relations.
Let $\psi$ be the unique minimizer of ${\mathcal E}[\psi]$, and denote by 
\begin{equation}
K = \int_{\mathbb{R}^3} (\nabla { \psi})^2 \, dx, 
\label{eq:HF4}
\end{equation}
\begin{equation}
A =  \int_{\mathbb{R}^3} \frac{Z}{|x|} \, { \psi}^2 \, dx, 
\label{eq:HF5}
\end{equation}
and, 
\begin{equation}
R =\frac{1}{2}  \int_{\mathbb{R}^3} \int_{\mathbb{R}^3} { \psi}^2(x) \frac{1}{|x-y|} {\psi}^2(y) \, dx \, dy.
\label{eq:HF6}
\end{equation}
Then we have the following identities.
\begin{theorem}[Virial Theorem]
If $\psi$ is the unique minimizer of ${\mathcal E}[\psi]$, and $K$, $A$ and $R$ are defined by (\ref{eq:HF4}), 
(\ref{eq:HF5}), and (\ref{eq:HF6}) respectively, then we have
\begin{equation}
2 K - A + R=0, 
\label{eq:HF7}
\end{equation}
\end{theorem}
\begin{proof}
Let $\psi_{\mu}(r) = \mu^{3/2} \, \psi(\mu \, r)$, then,  
\begin{equation} 
E(\mu) = {\mathcal E}[\psi_{\mu}] = \mu^2 \, K - \mu \, A + \mu \, R.
\label{eq:HF8}
\end{equation}
Because of the minimization property of $\psi$, 
\begin{equation} 
\frac{dE}{d\mu}\left(1\right) = 0, 
\label{eq:HF9}
\end{equation}
and therefore  (\ref{eq:HF7}) follows.
\end{proof}

\bigskip

Moreover we have the following relation. 
\begin{theorem}
If $\psi$ is the unique minimizer of ${\mathcal E}[\psi]$, and $K$, $A$ and $R$ are defined by (\ref{eq:HF4}), 
(\ref{eq:HF5}), and (\ref{eq:HF6}) respectively, then we have
\begin{equation}
K-A+2R=0.
\label{eq:HF10}
\end{equation}
\end{theorem}
\begin{proof} 
Multiply (\ref{eq:HF2}) by $\psi(x)$ and integrate over $\mathbb{R}^3$. Integrating by parts the left side and using the definition of $K$, $A$, and $R$, 
(\ref{eq:HF10}) follows.
\end{proof}

\bigskip
\bigskip
It follows from (\ref{eq:HF7}) and (\ref{eq:HF10}) that $3K=A$, i.e., 
if $\psi$ satisfies the Hartree equation (\ref{eq:HF2}) 
one has
\begin{equation}
\int_{\mathbb{R}^3} (\nabla { \psi})^2 \, dx = \frac{1}{3} \int_{\mathbb{R}^3} \frac{Z}{|x|} \, { \psi}^2 \, dx.
\label{eq:HF11}
\end{equation}
A key inequality to estimate $A$ is the well known {\it Coulomb Uncertainty Principle}. 
\begin{theorem}
For any $\psi \in H^1(\mathbb{R}^3)$ one has, 
\begin{equation}
\int_{\mathbb{R}^3} \frac{1}{|x|} \, { \psi(x)}^2 \, dx \le \|\nabla \psi \|_2 \|\psi\|_2, 
\label{eq:CUP}
\end{equation}
with equality if and only if $\psi(x) = B \, e^{-c |x|}$ for any constants $B$ and $c>0$.
\end{theorem}
\noindent
For a proof see, e.g., \cite{Los005}, Theorem 1, p. 14, or  \cite{LiSe010}, Equation (2.2.18), p.  29.

\bigskip
\bigskip
\noindent
In fact, we have, 
\begin{lemma} [An upper bound on $A$]
If $\psi \in H^1(\mathbb{R}^3)$ is the unique minimizer of (\ref{eq:HF1})
 (i.e., $\psi$  is the positive solution of the Hartree equation (\ref{eq:HF2})), 
one has, 
\begin{equation}
A \le \frac{1}{3} \, N \, Z^2.
\label{eq:HF11a}
\end{equation}
\end{lemma}
\begin{proof}
Using (\ref{eq:HF11}) and (\ref{eq:CUP}) one has, 
\begin{equation} 
A = 3 \,  \|\nabla \psi \|_2^2 \ge \frac{3}{Z^2} \, A^2 \, \frac{1}{N}, 
\label{eq:HF11b}
\end{equation}
and from here (\ref{eq:HF11a}) immediately follows.
\end{proof}
To conclude this section, we will prove an estimate on 
\begin{equation}
J \equiv  \int_{\mathbb{R}^3} |x| \, { \psi}^2 \, dx, 
\label{eq:HF11c}
\end{equation}
where $\psi$ is the solution to the Hartree equation (\ref{eq:HF2}).
\begin{lemma} [A lower bound on $J$]
If $\psi$ is the unique positive solution of the Hartree equation (\ref{eq:HF2}), 
one has, 
\begin{equation}
J \ge 3 \, \frac{N}{Z}.
\label{eq:HF11d}
\end{equation}
\end{lemma}
\begin{proof}
Using the Schwarz inequality, one has
 \begin{equation}
N^2  = \left(\int_{\mathbb{R}^3} \psi^2(x) \, dx\right)^2 \le \left(\int_{\mathbb{R}^3} |x| \, \psi^2(x) \, dx \right) \left(\int_{\mathbb{R}^3} \frac{1}{|x|}  \, \psi^2(x) \, dx \right),
\label{eq:HF11e}
\end{equation}
i.e., 
 \begin{equation}
N^2  \le J \, \frac{A}{Z} \le \frac{1}{3} \, J \, N\, Z,
\label{eq:HF11f}
\end{equation}
where we used (\ref{eq:HF11a}) to get the last inequality in (\ref{eq:HF11f}).
Finallly the lemma follows from (\ref{eq:HF11f}).

\end{proof}

\bigskip
\bigskip

\section{Upper bound on the critical charge for the Hartree equation. }
 
 \bigskip
 In this section we prove the main result of our manuscript namely, 
 
 \begin{theorem}[Upper bound on $N$]
 If $\psi \in H^1(\mathbb{R}^3)$ is the  unique positive solution to the Hartree equation (\ref{eq:HF2}), we have
 \begin{equation}
 N \le \frac{5}{4 \beta} Z \le   \frac{5}{4*0.8218} Z \approx 1.5211 \,  Z,
 \label{eq:3HF1}
 \end{equation}
 where $\beta$ is given by (\ref{eq:3HF8}) below.
 \end{theorem}

\begin{remarks}
\noindent
i) The proof that $N>Z$ is given in \cite{BeBrLi981}, Lemma 13, or in \cite{Lie981}, Theorem 7.16. In both cases take $p=5/3$ and 
$\gamma=0$. 

\noindent
ii) It is also known that $N<2 \, Z$ (see the comments and references  immediately below).

\noindent
iii) B.~Baumgartner. (see,  \cite{Bau984}, Section 4) computed numerically that $N \approx 1.21 Z$.

\end{remarks}

\bigskip
\bigskip
Before we go into the proof of Theorem 3.1, we recall that using the Benguria Lieb strategy one can prove the upper bound,
\begin{equation}
N \le 2 \, Z.
\label{eq:3HF2}
\end{equation}
For completeness, we recall the proof of (\ref{eq:3HF2}) (see, \cite{Lie981}, Theorem 7.22, p. 633, for details).
Multiplying (\ref{eq:HF2}) by $|x| \, \psi (x)$ and integrating over $\mathbb{R}^3$, we get
\begin{equation}
 \int_{\mathbb{R}^3} \left(-|x| \psi (x) \Delta \psi \right) \, dx  =  Z \,  \int_{\mathbb{R}^3} \psi^2(x) \, dx -  \int_{\mathbb{R}^3} \int_{\mathbb{R}^3} \psi^2(x) \frac{|x|}{|x-y|} \psi^2(y) \, dx \, dy
\label{eq:3HF3}
\end{equation}
Symmetrizing the second term in (\ref{eq:3HF3}), and using the triangular inequality we get, 
\begin{equation}
\int_{\mathbb{R}^3} \int_{\mathbb{R}^3} \psi^2(x) \frac{|x|}{|x-y|} \psi^2(y) \, dx \, dy = \frac{1}{2} \int_{\mathbb{R}^3} \int_{\mathbb{R}^3} \psi^2(x) \frac{|x|+|y|}{|x-y|} \psi^2(y) \, dx \, dy \ge 
\frac{1}{2} N^2, 
\label{eq:3HF4}
\end{equation}
where, as before, $N \equiv \int_{\mathbb{R}^3} \psi^2(x) \, dx$.
One can prove that 
\begin{equation}
\int_{\mathbb{R}^3} \left(-|x| \psi (x) \Delta \psi \right) \, dx  \ge 0.
\label{eq:3HF5}
\end{equation}
(see, \cite{Lie981}, or \cite{Lie984a,Lie984b}).
Finally from (\ref{eq:3HF3}), (\ref{eq:3HF4}) and (\ref{eq:3HF5}), the bound (\ref{eq:3HF2}) follows.

\bigskip
\bigskip
Now, using the strategy introduced by Nam in \cite{Nam012} (see also \cite{Nam013, Nam020}) we prove Theorem 3.1.

\bigskip
\begin{proof} [Proof of Theorem 3.1]

Multiplying this time (\ref{eq:HF2}) by $|x|^2 \, \psi (x)$, integrating over $\mathbb{R}^3$, and symmetrizing as before, we get
\begin{equation}
\int_{\mathbb{R}^3} \left(-|x|^2 \psi (x) \Delta \psi \right) \, dx  =  Z \,  \int_{\mathbb{R}^3} |x| \psi^2(x) \, dx - \frac{1}{2} \int_{\mathbb{R}^3} \int_{\mathbb{R}^3} \psi^2(x) \frac{|x|^2+|y|^2}{|x-y|} \psi^2(y) \, dx \, dy
\label{eq:3HF6}
\end{equation}
In this case, the integral on the left of (\ref{eq:3HF6}) is not non--negative. However, we can use the fact that for any real $f\in H^1(\mathbb{R}^3)$, one has that
\begin{equation}
\left(x^2 \, f, -\Delta f \right) \ge -\frac{3}{4} \left(f,f\right),
\label{eq:3HF7}
\end{equation}
(see, e.g. \cite{Nam012}, pp. 431, equation (9)) to bound the left side from below by $-3 \, N/4$. 

\bigskip
\bigskip

\noindent
Following Nam \cite{Nam012} we define, 
\begin{equation}
\beta = \inf \frac{1}{2} \left( \frac{\int_{\mathbb{R}^3} \int_{\mathbb{R}^3} \psi^2(x) \frac{|x|^2+|y|^2}{|x-y|} \psi^2(y) \, dx \, dy}{(\int_{\mathbb{R}^3} |x| \, \psi^2(x) \, dx)( \int_{\mathbb{R}^3}  \psi^2(x) \, dx)} \right)
\label{eq:3HF8}
\end{equation}
where the infimum is taken over all $\psi$, such that  $\int_{\mathbb{R}^3}  \psi^2(x) \, dx <\infty$. The exact numerical value of $\beta$
is not known, however, $0.8218 \le \beta \le 0.8705$  (\cite{Nam012}, Proposition 1).

\bigskip
\bigskip
\noindent
It follows from (\ref{eq:3HF8}) that 
\begin{equation}
 \frac{1}{2} \int_{\mathbb{R}^3} \int_{\mathbb{R}^3} \psi^2(x) \frac{|x|^2+|y|^2}{|x-y|} \psi^2(y) \, dx \, dy \ge \beta \, \int_{\mathbb{R}^3} |x| \, \psi^2 \, dx \, \int_{\mathbb{R}^3}  \psi^2 \, dx= \beta \, J \, N.
\label{eq:3HF9}
\end{equation}
where we have used (\ref{eq:HF11c}).
It follows from (\ref{eq:3HF6}), (\ref{eq:3HF7}), and (\ref{eq:3HF9}) that
\begin{equation}
\beta \, J \, N \le Z \, J + \frac{3}{4} N \le Z \, J+ \frac{1}{4} Z \, J, 
\label{eq:3HF10}
\end{equation}
where the last inequality in (\ref{eq:3HF10}) follows from (\ref{eq:HF11d}).
Finally, dividing both sides of  (\ref{eq:3HF10})  by $J$, and using the lower bound $\beta \ge 0.8218$ (see, \cite{Nam012}) the Theorem follows.
\end{proof}

\section*{Acknowledgments}
\thanks{This work has been supported by Fondecyt (Chile) Project \# 120--1055. 
One of us, TT, thanks the Instituto de F\ii sica of the Pontificia Universidad Cat\'olica de Chile 
for their support through a {\it Summer Research Fellowship}.


\end{document}